\documentclass[a4paper,12pt]{article}

\usepackage{amssymb,,amsmath}
\usepackage{graphicx}

\newtheorem{theorem}{Theorem}[section]
\newtheorem{proposition}[theorem]{Proposition}

\newtheorem{lemma}[theorem]{Lemma}
\newtheorem{corollary}[theorem]{Corollary}

\newenvironment{proof}{{\bf Proof }}{\hfill $\Box$}

\newcommand{\CC}{\mathbb{C}}

\newcommand{\EE}{\mathbb{E}}

\newcommand{\ZZ}{\mathbb{Z}}
\newcommand{\PP}{\mathbb{P}}
\def\ps#1#2{\langle\,{#1}\,,\,{#2}\,\rangle}

\newcommand{\rE}{\mathcal{E}}
\newcommand{\rH}{\mathcal{H}}
\newcommand{\rK}{\mathcal{K}}

\newcommand{\rM}{\mathcal{M}}

\newcommand{\rL}{\mathcal{L}}
\newcommand{\rV}{\mathcal{V}}

\newcommand{\s}{\sigma}
\font\timesept=cmr7

\def\tr{{\mbox{tr}}}

\def\qq{\qquad}

\def\sm{{\scriptstyle -}}

\def\tr{\mathop{\rm Tr\,}\nolimits}

\def\ps#1#2{\langle #1\, ,\, #2\rangle}

\def\norme#1{\left\| #1\right\|}
\def\normca#1{{\left\| #1\right\|}^2}
\def\ab#1{\left\vert #1\right\vert}

\def\D{\Delta}
\def\a{\alpha}
\def\b{\beta}
\def\s{\sigma}
\def\m{\mu}
\def\r{\rho}

\def\l{\lambda}
\def\g{\gamma}

\begin{document}


\title{Open\\Quantum Random Walks\footnote{Work supported by ANR project ``HAM-MARK", N${}^\circ$ ANR-09-BLAN-0098-01, by South African Research Chair Initiative of the Department of Science and Technology and National Research Foundation}}

\author{S. Attal${}^{{}^1}$, F. Petruccione${}^{{}^2}$, C. Sabot${}^{{}^1}$ and I. Sinayskiy${}^{{}^2}$}

\date{}

\maketitle

\centerline{\timesept ${}^{{}_1}$ Universit\'e de Lyon}
\vskip -1mm
\centerline{\timesept Universit\'e de Lyon 1, C.N.R.S.}
\vskip -1mm
\centerline{\timesept Institut Camille Jordan}
\vskip -1mm
\centerline{\timesept 21 av Claude Bernard}
\vskip -1mm
\centerline{\timesept 69622 Villeubanne cedex, France}

\bigskip
\centerline{\timesept ${}^{{}_2}$ Quantum Research Group}
\vskip -1mm
\centerline{\timesept School of Physics and National Institute for Theoretical Physics}
\vskip -1mm
\centerline{\timesept University of KwaZulu-Natal}
\vskip -1mm
\centerline{\timesept Durban 4001, South Africa}

\bigskip
\begin{abstract}
A new model of quantum random walks is introduced, on lattices as well as on finite graphs. These quantum random walks take into account the behavior of open quantum systems. They are the exact quantum analogues of classical Markov chains. We explore the ``quantum trajectory" point of view on these quantum random walks, that is, we show that measuring the position of the particle after each time-step gives rise to a classical Markov chain, on the lattice times the state space of the particle. This quantum trajectory is a simulation of the master equation of the quantum random walk. The physical pertinence of such quantum random walks and the way they can be concretely realized is discussed. Differences and connections with the already well-known quantum random walks, such as the Hadamard random walk, are established. 
\end{abstract}

\section{Introduction}

Random walks \cite{pt, rwp} are a useful mathematical concept, which found successful applications, e. g.,  in physics \cite{rwp}, computer science \cite{rwcs}, economics \cite{rwe} and biology \cite{rwb}. Basically, the trajectory of a random walk consists of a sequence of random steps on some underlying set of connected vertices \cite{rwp}. It is appealing to extend the concept of the classical random walk to the quantum domain. Quantum walks can be introduced in a discrete time \cite{aharonov} and in a continuous time \cite{FG} fashion. While for a classical random walk the probability distribution of the position of the walker depends only on the transition rates between the nodes of the graph, in the quantum case \cite{kempe} the probability amplitude of the walker depends on the dynamics of his internal degrees of freedom.  The appearance of interference effects makes these walks truly quantum. These quantum random walks, that we shall call Unitary Quantum Walks (for a reason which will appear clear in Section 10) have been successful for they give rise to strange behaviors of the probability distribution as time goes to infinity. In particular one can prove that they satisfy a rather surprising Central Limit Theorem whose speed is $n$, instead of $\sqrt{n}$ as usually, and the limit distribution is not Gaussian, but more like functions of the form (see \cite{Kon}): 
$$
x\mapsto\frac{\sqrt{1-a^2}\,(1-\l x)}{\pi\,(1-x^2)\,\sqrt{a^2-x^2}}\,,
$$
where $a$ and $\l$ are constants. 

Unitary quantum walks found wide application in quantum computing \cite{qaqrw}. Although, the physical implementation of  any quantum concept is usually  difficult due to unavoidable dissipation and decoherence effects \cite{toqs},  experimental realizations of unitary quantum random walks have been reported. Implementations with negligible effect of decoherence and dissipation were realized in optical lattices \cite{QWOL}, on photons in waveguide lattices \cite{qwwl}, with trapped ions \cite{qwti} and free single photons in space \cite{qwaw}. 

Recently, there has been interest in understanding the role of quantum transport in biological systems \cite{qbt}. Naturally, this raises the question of finding a framework for quantum walks in an open environment, for which dissipation and decoherence will play a non negligible role. On the contrary, such open quantum walks may even assist in the understanding of quantum efficiency. 

Over the last decade the implications of an open system approach to the dynamics of quantum walks have been addressed \cite{ken1}. Usually in these approaches the amount of dissipation and decoherence is minimal compared to the unitary driven coherence in the system of interest. In particular, effects of small amounts of decoherence have been shown to enhance some properties of quantum walks that make them useful for quantum computing. Recently, the framework of quantum stochastic walks was proposed \cite{qsw}, that allows to study the direct transition between classical random and quantum walks. Recently this transition was observed in various experiments \cite{exp}. General quantum walks on a lattice have been shown to have an interesting asymptotic behavior \cite{werner}.

The purpose of this article is to introduce a formalism for discrete time open quantum walks, which is exclusively based on the non-unitary dynamics induced by the local environments. The formalism suggested is similar to the formalism of quantum Markov chains \cite{gudder} and rests upon the implementation of appropriate completely positive maps \cite{toqs, kraus}. Our approach is rather different from \cite{gudder} for we put the emphasis on the quantum random walk character of these quantum Markov chains, we study their properties with the random walk point of view (limit distribution, transport properties, etc), we study their physical pertinence and make a physical connection with the unitary quantum random walks.

As we will show below the formalism of the open quantum random walks includes the classical random walk and through a physical realization procedure a connection to the unitary quantum walk is established. Furthermore, the OQRW allows for an unravelling in terms of quantum trajectories. In general, the behavior of the walk can not be explained in terms of classical or unitary walks. The particular properties of the OQRW make it a promising candidate for modeling of  quantum efficiency in biological systems and quantum computing.

\section{General Setup}\label{S:setup}

We now introduce the general mathematical and physical setup of the Open Quantum Random Walks. For sake of completeness we recall in this section several technical lemmas which ensure that our definitions are consistent. We omit their proofs as they all consist in easy exercises of Analysis. 

\bigskip
We are given a set $\rV$ of vertices, which might be finite or countable infinite. We consider all the oriented edges $\{(i,j)\,;\ i,j\in\rV\}$. We wish to give a quantum analogue of a random walk on the associated graph (or lattice).

We consider the space $\rK=\CC^\rV$, that is, the state space of a quantum system with as many degrees of freedom as the number of vertices; when $\rV$ is infinite countable we put $\rK$ to be any separable Hilbert space with an orthonormal basis indexed by $\rV$. We fix an orthonormal basis of $\rK$ which we shall denote by ${(\vert i\rangle)}_{i\in\rV}$. 

Let $\rH$ be a separable Hilbert space; it 
 stands for the space of degrees of freedom (or \emph{chirality} as they call it in Quantum Information Theory \cite{kempe}) given at each point of $\rV$. Consider the space $\rH\otimes \rK$. 

\smallskip
For each edge $(i,j)$ we are given a bounded operator $B^i_j$ on $\rH$. This operator stands for the effect of passing from $j$ to $i$. We assume that, for each $j$ 
\begin{equation}\label{E:rLi}
\sum_i {B^i_j}^* B^i_j= I\,,
\end{equation}
where the above series is strongly convergent (if infinite). This constraint has to be understood as follows: ``the sum of all the effects leaving the site $j$ is $I$\,". It is the same idea as the one for transition matrices associated to Markov chains: ``the sum of the probabilities leaving a site $j$ is 1". 

By Lemma \ref{L:1} which follows, to each $j\in\rV$ is associated a completely positive map on the density matrices of $\rH$:
$$
\rM_j(\rho)=\sum_i B^i_j \rho {B^i_j}^*\,.
$$
\begin{lemma}\label{L:1}
Let $(B_i)$ be a sequence of bounded operators on a separable Hilbert space $\rH$ such that the series
$\sum_i B_i^*B_i
$
converges strongly to a bounded operator $T$.
If $\r$ is a positive trace-class operator on $\rH$ then the series
$$
\sum_i B_i \r B_i^*
$$
is trace-norm convergent and 
$$
\tr\left(\sum_i B_i\r B_i^*\right)=\tr(\r T)\,.
$$
\end{lemma}

\bigskip
The operators $B^i_j$ act on $\rH$ only, we dilate them as operators on $\rH\otimes\rK$ by putting
$$
M^i_j=B^i_j\otimes \vert i\rangle\langle j\vert\,.
$$
The operator $M^i_j$ encodes exactly the idea that while passing from $\vert j\rangle$ to $\vert i\rangle$ on the lattice, the effect is the operator $B^i_j$ on $\rH$. 

By Lemma \ref{L:2}, which follows,  the series $\sum_{i,j} {M^i_j}^* M^i_j$ converges strongly to the operator $I$. 

\begin{lemma}\label{L:2}
Let $\rK$ and $\rH$ be separable Hilbert spaces. Consider an orthonormal basis $(\vert i\rangle)$ of $\rK$. Assume that  $B^i_j$ are bounded operators on $\rH$ such that, for all $j$, the series 
$
\sum_i {B^i_j}^*B^i_j
$
is strongly convergent to $I$. Define the bounded operators 
$$
M^i_j=B^i_j\otimes \vert i\rangle\langle j\vert
$$ 
on $\rH\otimes\rK$. Then the series
$
\sum_{i,j} {M^i_j}^*M^i_j
$
converges strongly to $I$.
\end{lemma}

\smallskip
As a consequence we can apply Lemma \ref{L:1} to the set of operators $(M^i_j)_{i,j}$ and the mapping 
\begin{equation}\label{E:rL}
\rM(\r)=\sum_i\sum_j M^i_j\,\r\, {M^i_j}^*
\end{equation}
defines a completely positive map on $\rH\otimes\rK$. 

\smallskip
We shall especially be interested in density matrices on $\rH\otimes\rK$ with the particular form
\begin{equation}\label{rho}
\r=\sum_i \r_i\otimes \vert i\rangle\langle i\vert\,,
\end{equation}
where each $\r_i$ is not exactly a density matrix on $\rH$: it is a positive and trace-class operator but its trace is not 1. Indeed the condition that $\r$ is a state aims to
\begin{equation}\label{sumri}
\sum_i \tr(\r_i)=1\,.
\end{equation} 
The importance of those density matrices is justified by the following.

\begin{proposition}\label{P:diagonal}
Whatever is the initial state $\r$ on $\rH\otimes\rK$, the density matrix $\rM(\r)$ is of the form (\ref{rho}).
\end{proposition}

Before proving this proposition, let us recall a basic result on partial traces.

\begin{lemma}\label{L:3}
Let $\r$ be a trace-class operator on $\rH\otimes\rK$ and $(\vert j\rangle)$ be an orthonormal basis of $\rK$. The operator 
$$
(I\otimes \vert i\rangle \langle j\vert)\,\r\,(I\otimes \vert j\rangle\langle i\vert)
$$
can be written as
$$
\r_j\otimes \vert i\rangle\langle i\vert
$$
for some trace-class operator $\r_j$ on $\rH$, which we shall denote by $\langle j\vert\,\r\,\vert j\rangle$. Furthermore we have
$$
\tr(\langle j\vert\,\r\,\vert j\rangle)=\tr\big(\r\, (I\otimes\vert j\rangle\langle j\vert)\big)\,.
$$
\end{lemma}

\bigskip
We can now come back to the proof of the proposition.\\
\begin{proof}[of Proposition \ref{P:diagonal}]
We have
\begin{align*}
\rM(\r)&=\sum_{i,j}\left(B^i_j \otimes \vert i\rangle\langle j\vert\right)\,\r\, \left({B^i_j}^*\otimes\vert j\rangle\langle i\vert\right)\\
&=\sum_{i,j} (B^i_j \otimes I)(I\otimes \vert i\rangle\langle j\vert)\,\r\, (I\otimes\vert j\rangle\langle i\vert)({B^i_j}^*\otimes I)\,.
\end{align*}
If we put $\r_j=\langle j\vert\,\r\,\vert j\rangle$ (as in Lemma \ref{L:3}), we get
\begin{align*}
\rM(\r)&=\sum_{i,j} (B^i_j \otimes I)(\r_j\otimes \vert i\rangle\langle i\vert)({B^i_j}^*\otimes I)\\
&=\sum_{i,j} B^i_j \r_j{B^i_j}^*\otimes \vert i\rangle\langle i\vert\,.
\end{align*}
Each of the operators $B^i_j \r_j{B^i_j}^*$ is positive and trace-class, hence so is the operator $\sum_{j\leq M} B^i_j \r_j{B^i_j}^*$. But we have
$$
\tr\left(\sum_{j\leq M} B^i_j \r_j{B^i_j}^*\right)=\sum_{j\leq M}\tr(\r_j{B^i_j}^*B^i_j )\,.
$$
As $\sum_i {B^i_j}^*B^i_j=I$, each of the operators ${B^i_j}^*B^i_j$ is smaller than $I$ (in the sense that $I-{B^i_j}^*B^i_j$ is a positive operator). Hence, $\tr(\r_j{B^i_j}^*B^i_j )\leq \tr(\r_j)$, as can be easily checked. This shows that 
$$
\sum_{j}\tr\left( B^i_j \r_j{B^i_j}^*\right)<\infty
$$
and that $\sum_{j\leq M} B^i_j \r_j{B^i_j}^*$ converges in trace-norm to a positive trace-class operator $\sum_{j} B^i_j \r_j{B^i_j}^*$ which satisfies
$$
\tr\left(\sum_{j} B^i_j \r_j{B^i_j}^*\right)=\sum_{j}\tr(\r_j{B^i_j}^*B^i_j )\,.
$$
In particular by Lemma \ref{L:1} we have, 
\begin{align*}
\sum_i\tr\left(\sum_{j} B^i_j \r_j{B^i_j}^*\right)&=\sum_i\sum_{j}\tr(\r_j{B^i_j}^*B^i_j )\\
&=\sum_j \tr\left(\r_j\left(\sum_i{B^i_j}^*B^i_j\right)\right)\\
&=\sum_j \tr(\r_j)\\
&=1\,.
\end{align*}
This means that the series (in the variable $i$)
$$
\sum_i\left(\sum_{j} B^i_j \r_j{B^i_j}^*\right)\otimes \vert i\rangle\langle i\vert
$$
is trace-norm convergent. We now immediately have the relation
$$
\sum_{i,j} B^i_j \r_j{B^i_j}^*\otimes \vert i\rangle\langle i\vert=\sum_i\left(\sum_{j} B^i_j \r_j{B^i_j}^*\right)\otimes \vert i\rangle\langle i\vert\,.
$$
This proves that $\rM(\r)$ is of the form (\ref{rho}).
\end{proof}

\bigskip
The states of the form (\ref{rho}) are mixtures of initial states $\r_i$ on each site $i$, but  they express no mixing between the sites. An immediate consequence of the proof of Proposition \ref{P:diagonal} is the following important formula.

\begin{corollary}\label{C:Lrrho}
If $\r$ is a state on $\rH\otimes\rK$ of the form
$$
\r=\sum_i \r_i\otimes \vert i\rangle\langle i\vert\,,
$$
then
\begin{equation}\label{E:rLrho}
\rM(\r)=\sum_i \left(\sum_j B^i_j \r_j {B^i_j}^*\right)\otimes \vert iÊ\rangle\langle i\vert\,.
\end{equation}
\end{corollary}
This is exactly the quantum analogue of a usual random walk: after one step, on the site $i$ we have all the contributions from those pieces of the state which have travelled from $j$ to $i$.

\section{Open Quantum Random Walks}
If the state of the system $\rH\otimes\rK$ is of the form 
$$
\r=\sum_i \r_i\otimes\vert i\rangle\langle i\vert\,,
$$
then a measurement of the ``position" in $\rK$, that is, a measurement along the orthonormal basis $(\vert i\rangle)_{i\in\rV}$, gives the value $\vert i\rangle$ with probability 
$$
\tr(\r_i)\,.
$$
As proved in Corollary \ref{C:Lrrho}, after applying the completely positive map $\rM$ the state of the system $\rH\otimes\rK$ is
$$
\rM(\r)=\sum_{i}\sum_j B^i_j\, \r_j \,{B^i_j}^*\otimes \vert i\rangle\langle i\vert\,.
$$
Hence, a measurement of the position in $\rK$ would give that each site $i$ is occupied with probability
\begin{equation}\label{E:pi}
\sum_j \tr\left(B^i_j \r_j {B^i_j}^*\right)\,.
\end{equation}

\smallskip
Now, let us see what happens if the measurement is performed after two steps only. In this case, the state of the system is
$$
\rM^2(\r)=\sum_{i}\sum_j\sum_k B^i_j B^j_k\, \r_k\, {B^j_k}^*{B^i_j}^*\otimes \vert i\rangle\langle i\vert\,.
$$
Hence measuring the position, we get the site $\vert i\rangle$ with probability
\begin{equation}\label{E:L2}
\sum_j\sum_k \tr\left(B^i_jB^j_k \,\r_k \,{B^j_k}^*{B^i_j}^*\right)\,.
\end{equation}
Clearly, there is no way to understand the probability measure given in (\ref{E:L2}) for two steps with the help of only the probability measure on one step (\ref{E:pi}). One needs to know the density matrices $\r_k$.

\smallskip
The random walk which is described this way by the iteration of the completely positive map $\rM$ is not a classical random walk, it is a quantum random walk. The rules for jumping from a site to another are dictated by the sites, but also by the chirality. This is what we call an ``Open Quantum Random Walk". 

Let us resume these remarks in the following proposition, which follows easily from the previous results and remarks.

\begin{proposition}\label{P:DQRW}
Given any initial state $\r^{(0)}$ on $\rH\otimes\rK$, then for all $n\geq 1$ the states $\r^{(n)}=\rM^n(\r^{(0)})$ are all of the form
$$
\r^{(n)}=\sum_i \r^{(n)}_i\otimes \vert i\rangle\langle i\vert\,.
$$
They are given inductively by the following relation:
$$
\r^{(n+1)}_i=\sum_j B^i_j\, \r^{(n)}_j\,{B^i_j}^*\,.
$$
For each $n\geq 1$, the quantities 
$$
p^{(n)}_i=\tr(\r^{(n)}_i)\,,\  \ i\in\rV
$$
define a probability distribution $p^{(n)}$ on $\rV$, it is called the ``\emph{probability distribution of the open quantum random walk at time $n$}". 
\end{proposition}

Before going ahead with the properties of the Open Quantum Random Walks, let us introduce a few examples.

\section{Examples on $\ZZ$}\label{S:exampleZ}

It is very easy to define a stationary open quantum random walk on $\ZZ$. Let $\rH$ be any Hilbert space and $B,C$ be two bounded operators on $\rH$ such that 
$$
B^*B+C^*C=I\,.
$$
Then we can define an open quantum random walk on $\ZZ$ by saying that one can only jump to nearest neighbors: a jump to the left is given by $B$ and a jump to the right is given by $C$.
In other words, we put
$$
B^{i-1}_i=B
\qquad\mbox{and}\qquad
B^{i+1}_i=C
$$
for all $i\in\ZZ$, all the others $B^i_j$ being equal to 0.

\smallskip
Starting with an initial state $\r^{(0)}=\r_0\otimes \vert 0\rangle\langle 0\vert$, after one step we have the state
$$
\r^{(1)}=B\r_0B^*\otimes \vert \sm1\rangle\langle \sm1\vert+C\r_0C^*\otimes \vert 1\rangle\langle 1\vert\,.
$$
The probability of presence in $\vert\sm1\rangle $ is $\tr(B\r_0 B^*)$ and the probability of presence in $\vert 1\rangle$ is $\tr(C\r_0C^*)$. 

After the second step, the state of the system is
\begin{align*}
\r^{(2)}&=B^2\r_0{B^2}^*\otimes \vert\sm2\rangle\langle \sm2\vert+C^2\r_0{C^2}^*\otimes \vert 2\rangle\langle 2\vert+\\
&\ \ \ +\left(CB\r_0B^*C^*+BC\r_0C^*B^*\right)\otimes \vert 0\rangle\langle 0\vert\,.
\end{align*}
The associated probabilities for the presence in $\vert \sm2\rangle$, $\vert0\rangle$, $\vert 2\rangle$ are then 
$$
\tr(B^2\r_0{B^2}^*),\ \ \ \tr(CB\r_0B^*C^*+BC\r_0C^*B^*)\ \ \mbox{and}\ \ \tr(C^2\r_0{C^2}^*)\,,
$$ 
respectively. 

One can iterate the above procedure and generate our open quantum random walk on $\ZZ$.

\smallskip
As further example, take 
$$
B=\frac{1}{\sqrt3}\,\left(\begin{matrix}1&1\\0&1\end{matrix}\right)
\qquad\mbox{and}\qquad
C=\frac{1}{\sqrt3}\,\left(\begin{matrix}1&0\\-1&1\end{matrix}\right)\,.
$$
The operators $B$ and $C$ do satisfy $B^*B+C^*C=I$. Let us consider the associated open quantum random walk on $\ZZ$. Starting with the state 
$$
\r^{(0)}=\left(\begin{matrix}1&0\\0&0\end{matrix}\right)\otimes \vert0\rangle\langle0\vert\,,
$$
we find the following probabilities for the 4 first steps:
$$
\begin{matrix} 
&\vert-4\rangle&\vert-3\rangle&\vert-2\rangle&\vert-1\rangle&\vert0\rangle&\vert+1\rangle&\vert+2\rangle&\vert+3\rangle&\vert+4\rangle\\
n=0&&&&&1&&&&\\
n=1&&&&
\frac 13&&\frac 23&&&\\
n=2&&&
\frac 19&&\frac 39&& \frac 59&&\\
n=3
&&\frac 1{27}&& \frac 5{27}&&\frac{11}{27}&&\frac{10}{27}&\\
n=4&\frac 1{81}&&\frac{10}{81}&&\frac{27}{81}&&\frac{26}{81}&&\frac{17}{81}\\
\end{matrix} 
$$

\smallskip
The distribution obviously starts asymmetric, uncentered and rather wild. The interesting point is that, while keeping its quantum behavior time after time, simulations show up clearly a tendancy to converge to a normal centered distribution. 

A proof of this fact, in this particular example, is accessible making use of the translation invariance of the walk and using Fourier transform techniques, in the same way as as is done (though, in a quite different language) in \cite{werner}. It is to be noticed that such methods do not give explicit or  accessible parameters for the limit Gaussian distribution. Indeed, such methods are based on the spectral behavior near the origin of a particular perturbation of the dynamics; the associated parameters are then in general hard to compute. Another approach to these central limit behaviors is developed in \cite{AGS}, based on Markov chain and random media techniques; it gives rise to more general results than \cite{werner} and it gives rise to explicit parameters for the Gaussian limit distributions. 

\smallskip
Coming back to our example, Figure 1 below shows the distribution obtained at times $n=4$, $n=8$ and $n=20$.

\begin{figure}[h!]
\begin{center}
\leavevmode
{%
      \begin{minipage}{0.4\textwidth}
        \includegraphics[width=4cm,height=3.2cm]
        {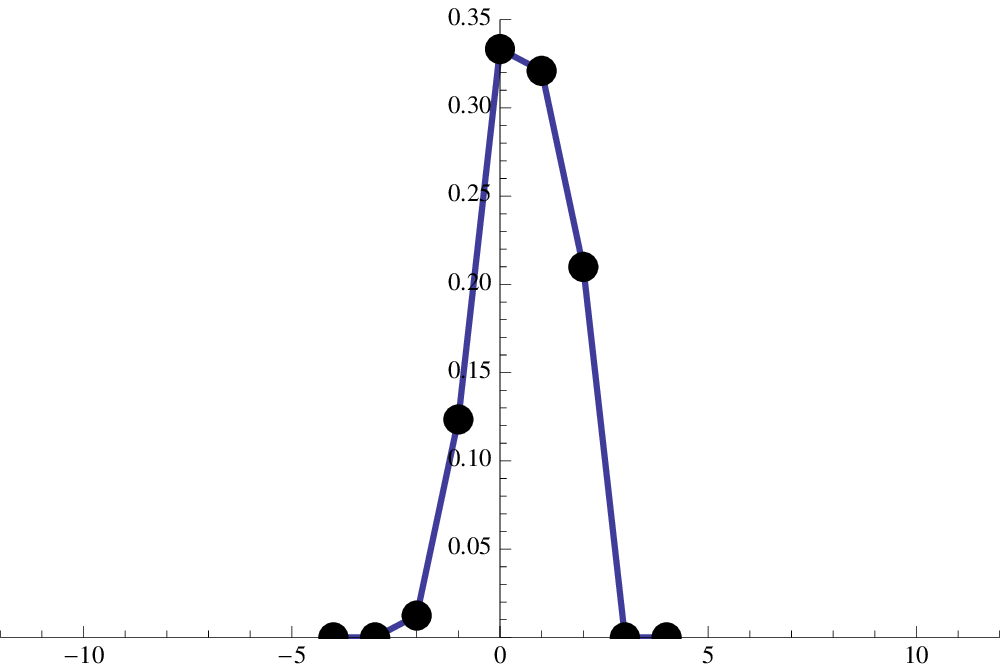}\\
     \vspace{-0.9cm} \strut
      
        \end{minipage}}\hspace*{-0.8cm}
{%
      \begin{minipage}{0.4\textwidth}
        \includegraphics[width=4cm,height=3.2cm]
         {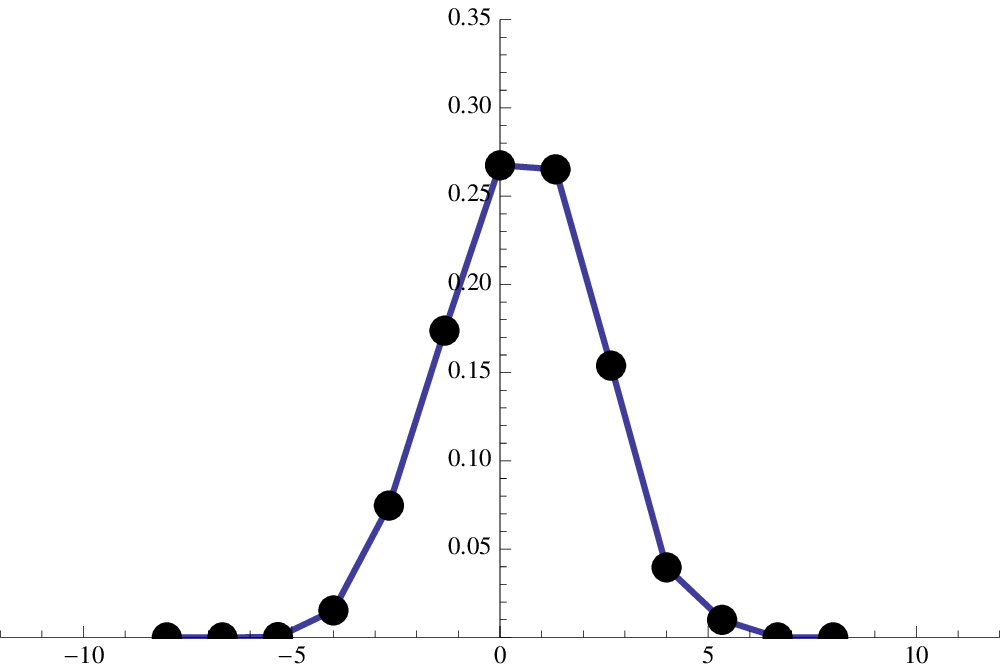}\\
        \vspace{-0.9cm} \strut
        \end{minipage}}\hspace*{-0.8cm}
{%
      \begin{minipage}{0.4\textwidth}
        \includegraphics[width=4cm,height=3.2cm]
        {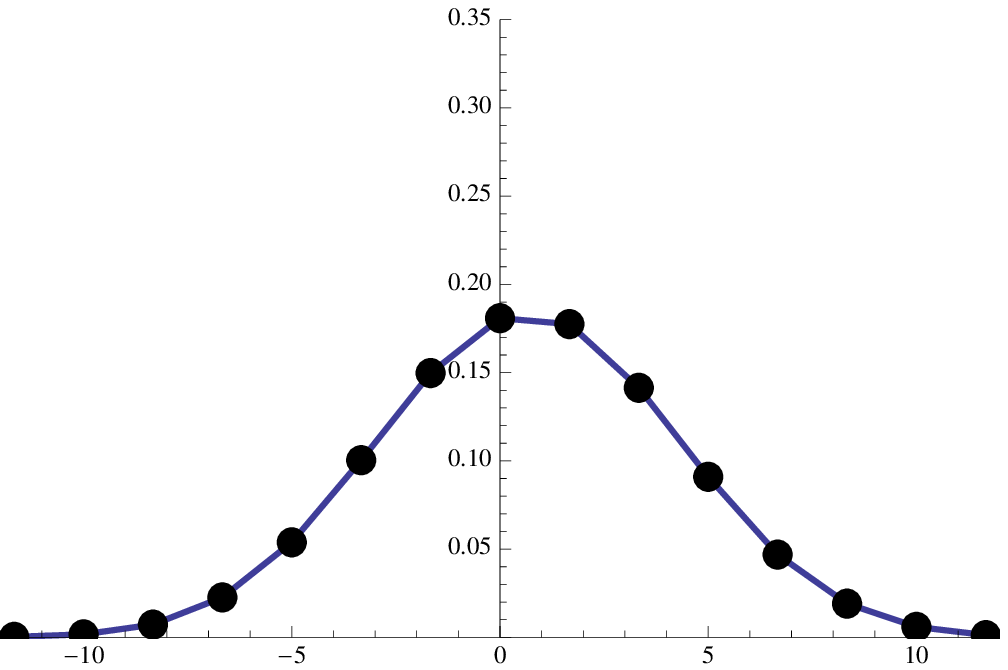}\\
      \vspace{-0.9cm} \strut
       \end{minipage}}
\end{center}
\vskip -0.1cm
\caption{\it An O.Q.R.W. on $\ZZ$ which gives rise to a centered Gaussian at the limit, while starting clearly uncentered (at time 
$n=4$, $n=8$, $n=20$)}
\end{figure} 

\bigskip
One can produce examples where several Gaussians are appearing, including the case where Gaussians are reduced to Dirac masses. It is to be noted that the rigorous proof of such asymptotic behaviors, where several Gaussians are appearing, is out of reach with simple arguments like Fourier transforms. In the article \cite{werner}, cited above, their techniques only allow to consider the case of a single Gaussian distribution at the limit. A detailed study of Central Limit Theorems for O.Q.R.W. is presented in \cite{AGS}, where some of the cases with several Gaussians are considered. A general Central Limit Theorem for O.Q.R.W. is, at the time we write, still an open problem. 
In this article, we stick to numerical simulations only and we refer to \cite{AGS} for a detailed study of central limit theorems in this context. 

\smallskip
For example, let us take $\rH$ being of dimension 5. For the sake of a compact notation, we put $C_2=\cos(2t)$, $C_4=\cos(4t)$, $S_2=\sin(2t)$ and $S_4=\sin(4t)$. Consider the matrices
$$
B=
\frac{1}{4}\left(
\begin{matrix}
 0 & -2 S_2-S_4 & 0 & 2 S_2-S_4 & 0 \\
-2 S_2-S_4 & 0 & -2\sqrt{\frac{3}{2}} S_4 & 0 & 2 S_2-S_4 \\
 0 & -2\sqrt{\frac{3}{2}} S_4 & 0 & -2\sqrt{\frac{3}{2}} S_4 & 0 \\
2 S_2-S_4 & 0 & -2\sqrt{\frac{3}{2}} S_4 & 0 & -2 S_2-S_4 \\
 0 & 2 S_2-S_4 & 0 & -2 S_2-S_4 & 0
\end{matrix}
\right)
$$
and
$$
C=
 \frac{1}{8}\left(
\begin{matrix}
 L & 0 & C & 0 &L' \\
 0 & 4(C_2+C_4) & 0 & 4(-C_2+C_4) & 0 \\
C & 0 & 2(1+3 C_4) & 0 & C \\
 0 & 4 (-C_2+C_4) & 0 & 4(C_2+C_4) & 0 \\
 L' & 0 & C & 0 & L
\end{matrix}
\right)\,,
$$
where
$$
L=3+4 C_2+C_4,\qq L'=3-4 C_2+C_4,\qq C=-\sqrt{6}(1-C_4)\,.
$$
Simulations of this open quantum random walk indicates that the limit behavior exhibits two Gaussians plus a Dirac soliton. The two Gaussians get slowly constructed, point by point, as the soliton loses its mass. In Figure 2 we show the time evolution when the parameter $t$ is equal to $t=\pi/40$, the initial state being 
$
\r^{(0)}=\frac 15 I\otimes\vert0\rangle\langle 0\vert\,.
$
\begin{figure}[h!]
\begin{center}
\leavevmode
{%
      \begin{minipage}{0.4\textwidth}
        \includegraphics[width=4cm,height=3.2cm]
        {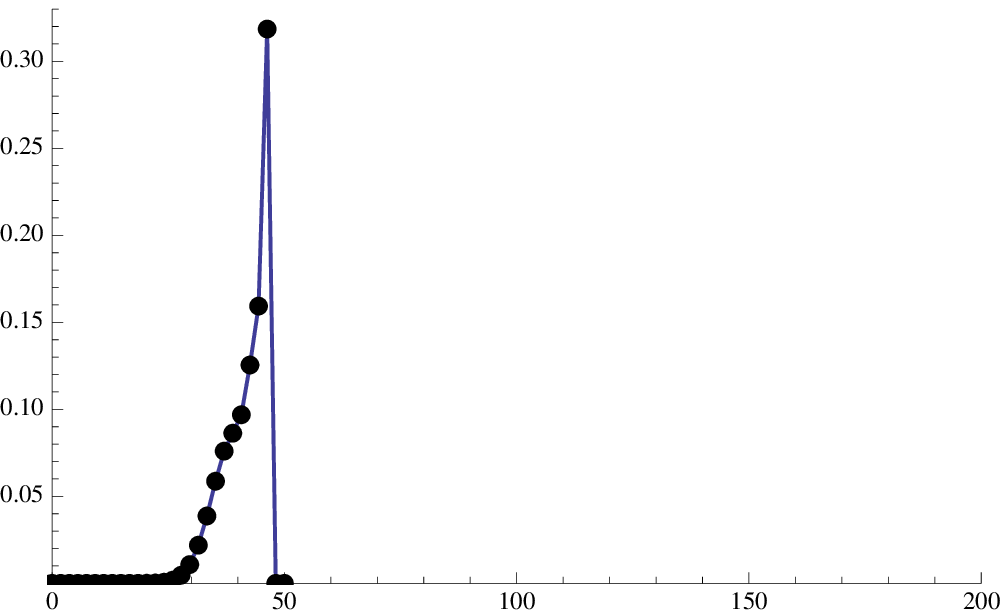}\\
     \vspace{-0.9cm} \strut
      
        \end{minipage}}\hspace*{-0.8cm}
{%
      \begin{minipage}{0.4\textwidth}
        \includegraphics[width=4cm,height=3.2cm]
         {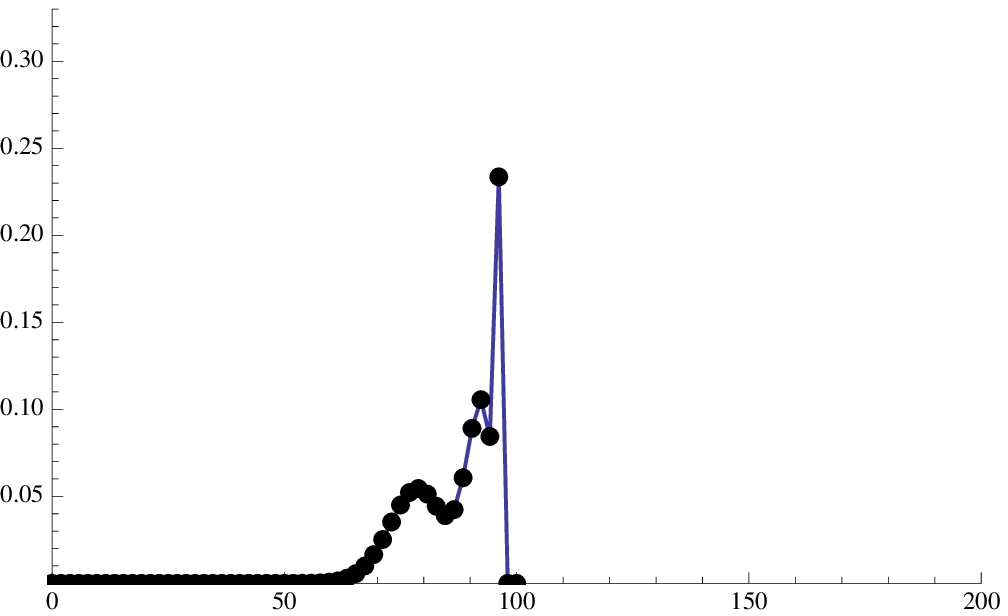}\\
        \vspace{-0.9cm} \strut
        \end{minipage}}\hspace*{-0.8cm}
{%
      \begin{minipage}{0.4\textwidth}
        \includegraphics[width=4cm,height=3.2cm]
        {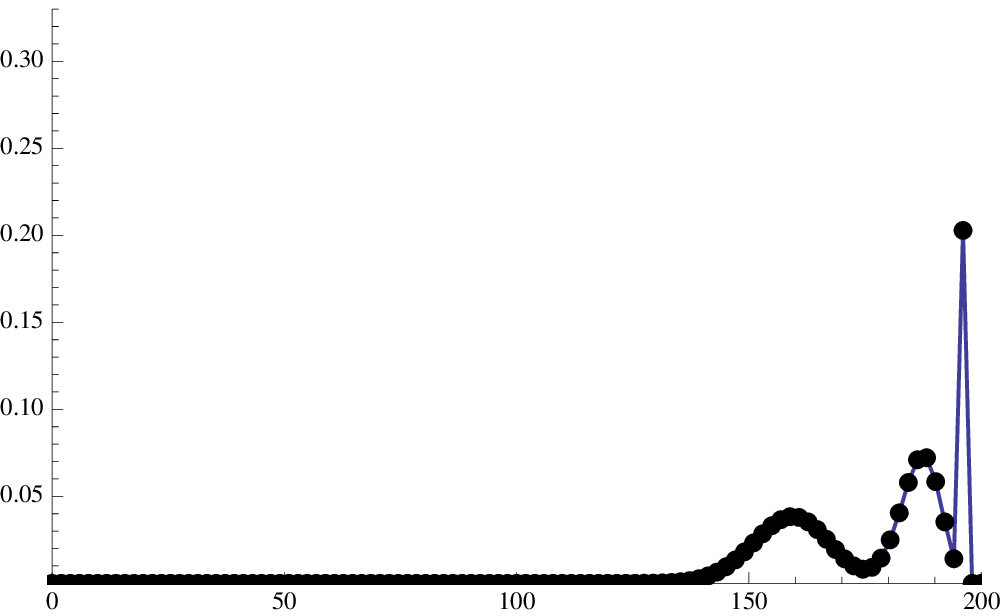}\\
      \vspace{-0.9cm} \strut
        \end{minipage}}
\end{center}
\vskip -0.1cm
\caption{\it Two Gaussians moving towards right are constructed point by point while the soliton loses its mass (parameter $t=\pi/40$, times $n=50$, $n=100$, $n=200$).}
\end{figure} 
Changing the parameter $t$ makes the Gaussians moving at different speeds and even change their direction.

\section{Examples on Graphs}\label{S:example_graph}

In order to give examples on finite graphs it is useful to fix a notation. We shall denote the operators involved in the open quantum random walk in a way similar to the notation of stochastic matrices for Markov chains. If the set of vertices is $\rV=\{1,\ldots,V\}$, we shall denote the operators $B^i_j$ inside a $V\times V$-matrix as follows:
$$
\left(
\begin{matrix}
B^1_1&B^2_1&\ldots&B^V_1\\
B^1_2&B^2_2&\ldots&B^V_2\\
\vdots&\vdots&\vdots&\vdots\\
B^1_V&B^2_V&\ldots&B^V_V
\end{matrix}
\right)\,.
$$
That is, on line $j$ are all the operators for the contributions $B^i_j$ which start from $j$ and go to another site $i$. The usual property for stochastic matrices that the sum of each line is 1, is replaced by 
$
\sum_i {B^i_j}^*B^i_j=I
$
for each line. 
With this notation one can easily describe examples. On the graph with two vertices we consider the transition operators of the form
$$
\left(
\begin{matrix}
D_1&D_2\\
B&C\\
\end{matrix}
\right)
$$
where $D_1$ and $D_2$ are any diagonal matrices such that $D_1^*D_1+D_2^*D_2=I$ and where 
$$
B=\left(\begin{matrix}0&\sqrt p\\0&0\end{matrix}\right)\qquad \mbox{and}\qquad C=\left(\begin{matrix} 1&0\\0&\sqrt{1-p}\end{matrix}\right)\,,
$$
for some $p\in(0,1)$.
It is easy to prove rigorously that the density matrix of these examples always converges to 
$$
\left(\begin{matrix}1&0\\0&0\end{matrix}\right)\otimes\vert 2\rangle\langle 2\vert\,.
$$
This is to say that the ``mass'' is always asymptotically leaving the site $\vert 1\rangle$ in order to ``charge'' the site $\vert 2\rangle$ only. 

The above idea can be pushed further to a chain of $N$ sites connected as follows
$$
\left(\begin{matrix}D_1&D_2&0&0&\ldots&0&0&0\\
D_3&0&D_4&0&\ldots&0&0&0\\
0&D_5&0&D_6&\ldots&0&0&0\\
\vdots&\ddots&\ddots&\ddots&\ldots&\vdots&\vdots&\vdots\\
0&0&0&0&\ldots&D_{2N-3}&0&D_{2N-2}\\
0&0&0&0&\ldots&0&B&C\end{matrix}\right)\,.
$$
Then any initial state, for example any state of the form $\r_0\otimes \vert 1\rangle\langle 1\vert$, will converge to the state
$$
\left(\begin{matrix}1&0\\0&0\end{matrix}\right)\otimes \vert N\rangle\langle N\vert\,.
$$
Though this example is rather classical in its behavior, it is interesting for it gives a model of a sort of ``excitation transport": giving any initial state on the site 1 only, the state will then be, more or less quickly, transported along the chain and will end up into the excited state on the site $\vert N\rangle$, see \cite{R-M} for example.

\bigskip
An important open problem is to be noted in this context. In the case of classical Markov chain on a finite (or even countable) graph, a classification of the behaviors is well-known. The nodes are separated into \emph{recurrent} and \emph{transient} ones. The definition of a recurrent or a transient node is rather simple and clear (a node is transient if, starting from it one can reach another node with strictly positive probability, but with no probability of coming back; if it is not transient a node is recurrent). It is well-known that the invariant measures are exactly supported by the recurrent nodes.

\smallskip
Regarding the few examples described above, it is natural to wonder if there exists some equivalent characterization for OQRW. How can one see on the $B_i^j$'s that a given node will be ``recurrent'', that is, will be in the support of the invariant states? For the moment, no obvious answer appeared to this question; the characterization of recurrence and transience seems not so simple in this quantum context. 

\section{Recovering Classical Markov Chains}
Let us now come back to the general setup of Open Quantum Random Walks.

\smallskip
It is very interesting to notice that all the classical Markov chains can be recovered as particular cases of Open Quantum Random Walks. We only treat here the homogenous case, the discussion would be similar in the non-homogeneous case.

\smallskip
Consider $P=(P(j,i))$ a stochastic matrix, that is $P(j,i)$ are classical probability transitions on $\rV$. They express the transition probabilities of a Markov chain $(X_n)$ on $\rV$, that is,
$$
P(j,i)=\PP(X_{n+1}=i\,\vert\, X_n=j)\,.
$$
In particular, recall that
$$
\sum_{i\in\rV} P(j,i)=1
$$
for all $j$. 

\begin{proposition}\label{P:classical}
Put $\rH=\rK=\CC^\rV$ and consider any family of unitary operators $U^i_j$ on $\CC^N$, $i,j\in\rV$. Consider the operators 
$$
B^i_j=\sqrt{P(j,i)}\, U^i_j\,.
$$
They satisfy
$$
\sum_{i} {B^i_j}^* B^i_j=I
$$
for all $j$. Furthermore, given any initial state $\r^{(0)}$,  the associated open quantum random walk $(\rM^n)$ has the same probability distributions $(p^{(n)})$ as the classical Markov chain $(X_n)$ with transition probability matrix $P$ and initial measure 
$$
p^{(0)}_i= \tr(\langle i\vert\, \r^{(0)}\,\vert i\rangle)\,.
$$
\end{proposition}
\begin{proof}
The relation on the operators $B^i_j$ is obvious for
$$
\sum_{i} {B^i_j}^* B^i_j=\sum_i P(j,i)\, {U^i_j}^* U^i_j=\sum_i P(j,i) I=I\,.
$$
Whatever is the initial state $\r$, if we put $\r_i=\langle i\vert\, \r\,\vert i\rangle$, we get by Proposition \ref{P:diagonal} and its proof
$$
\rM(\r)=\sum_i\left(\sum_k P(k,i) {U^i_k}\r_k{U^i_k}^*\right)\otimes \vert i\rangle\langle i\vert\,.
$$
The probability to be located on site $i$ is then
$$
\sum_k P(k,i) \tr\left({U^i_k}\r_k{U^i_k}^*\right)=\sum_k P(k,i) \tr(\r_k)\,.
$$
That is, we get the classical transition probabilities for a classical Markov chain on the set $\rV$, driven by the transition probabilities $P(i,j)$ and with initial measure $p^{(0)}_i=\tr(\r_i)$. 

After two steps, the probability to be located at site $i$ is
$$
\sum_k \sum_l P(l,k)P(k,i) \tr\left({U^i_k}U_l^k\r_l{U_l^k}^*{U^i_k}^*\right)=\sum_k\sum_l P(l,k) P(k,i) \tr(\r_l)\,.
$$
That is, once again the usual transition probabilities for two steps of the above Markov chain. It is not difficult to get convinced, by induction, that this works for any number of steps. 
\end{proof}

\section{Quantum Trajectories}
We shall now describe a very interesting and convenient way to simulate OQRWs by  means of \emph{Quantum Trajectories}. This property seems very important when one wants to study the limit behavior of these quantum random walks (see \cite{AGS}). 

\smallskip
The principle of the quantum trajectories associated to an open quantum random walk is the following. Starting from any initial state $\r$ on $\rH\otimes\rK$ we apply the mapping $\rM$ and then a measurement of the position in $\rK$. We end up with a random result for the measurement and a reduction of the wave-packet gives rise to a random state on $\rH\otimes\rK$ of the form
$$
\r_i\otimes \vert i\rangle\langle i\vert\,.
$$
We then apply the procedure again: an action of the mapping $\rM$ and a measurement of the position in $\rK$.

\begin{theorem}\label{T:quantum_traj}
By repeatedly applying the completely positive map $\rM$ and a measurement of the position on $\rK$, one obtains a sequence of random states on $\rH\otimes\rK$. This sequence is a non-homogenous Markov chain with law being described as follows. If the state of the chain at time $n$ is $\r\otimes \vert j\rangle\langle j\vert$, then at time $n+1$ it jumps to one of the values
$$
\frac{1}{p(i)}\,B^i_j \r {B^i_j}^*\otimes \vert i\rangle\langle i\vert\,,\ \ i\in\rV,
$$
with probability
$$
p(i)=\tr\left(B^i_j \r {B^i_j}^*\right)\,.
$$
This Markov chain $(\r^{(n)})$
is a simulation of the master equation driven by $\rM$, that is,
$$
\EE\left[\r^{(n+1)}\,\vert\, \r^{(n)}\right]=\rM(\r^{(n)})\,.
$$
Furthermore, if the initial state is a pure state, then the quantum trajectory stays valued in pure states and the Markov chain is described as follows. If the state of the chain at time $n$ is the pure state $\vert \varphi\rangle\otimes \vert j\rangle$, then at time $n+1$ it jumps to one of the values
$$
\frac{1}{\sqrt{p(i)}}\, B^i_j\, \vert\varphi\rangle \otimes \vert i\rangle\,,\ \ i\in\rV,
$$
with probability
$$
p(i)=\normca{B^i_j \,\vert\varphi\rangle}\,.
$$
\end{theorem}
\begin{proof}
Let $\r\otimes\vert j\rangle\langle j\vert$ be the initial state. After acting by $\rM$ the state is
$$
\sum_i (B_i^j \r {B_i^j}^*)\otimes \vert i\rangle\langle i\vert\,.
$$
Measuring the vertices, gives the site $i$ with probability
$$
p(i)=\tr(B_i^j \r {B_i^j}^*)\,.
$$
By the usual wave-packet reduction postulate, the state after having been measured with this value is
$$
\frac 1{p(i)} (B_i^j \r {B_i^j}^*)\otimes\vert i\rangle\langle i\vert\,.
$$
This state being given, if we repeat the procedure, then clearly the next step depends only on the new state of the system. We end up with a (non-homogenous) Markov chain structure. 

On average, the values of this Markov chain after one step is
\begin{align*}
\EE\left[\r^{(n+1)}\,\vert\, \r^{(n)}=\r\otimes \vert j\rangle\langle j\vert\right]&=\sum_i p(i)\,\frac 1{p(i)} (B_i^j \r {B_i^j}^*)\otimes\vert i\rangle\langle i\vert\\
&=\sum_i  (B_i^j \r {B_i^j}^*)\otimes\vert i\rangle\langle i\vert\\
&=\rM(\r^{(n)})\,.
\end{align*}

If $\r$ is a pure state $\vert \phi\rangle\langle\phi\vert\otimes \vert i\rangle\langle i\vert$, then it stays a pure state at each step. Indeed, any initial pure state $\vert \phi\rangle\langle \phi\vert\otimes\vert i\rangle\langle i\vert$ will jump randomly to one of the states 
$$
\frac 1{p^j_i}\,B^j_i\vert  \phi\rangle\langle \phi\vert {B^j_i}^*\otimes \vert j\rangle\langle j\vert
$$
with probability 
$$
p(i)=\tr(B^j_i \vert  \phi\rangle\langle \phi\vert {B^j_i}^*)\,.
$$
In other words, it jumps from the pure state $\vert \phi\rangle\otimes \vert i\rangle$ to any of the pure states 
$$
\frac 1{\sqrt{p(i)}}\,B^j_i\vert  \phi\rangle\otimes \vert j\rangle
$$ 
with probability
$$
p(i)=\norme{B^j_i \vert  \phi\rangle}^2\,.
$$
We have a classical Markov chain valued in the space of wave functions of the form $\vert \phi\rangle\otimes \vert i\rangle$. On average, this random walk simulates the master equation driven by $\rM$. 
\end{proof}

\section{Realization Procedure}

It is natural to wonder how such Open Quantum Random Walks can actually be realized physically. We shall here discuss a way to achieve it.

For the sake of a simple discussion, we restrict ourselves in this section to the case where either $\rV$ is finite, or the number of non-vanishing $B^i_j$'s is finite for every fixed $j$. This is the case in all our examples and makes all the sums finite in the following.

\smallskip
Consider an open quantum random walk on $\rV$ with chirality space $\rH$ and with associated transition operators $B^i_j$. Recall that we have supposed that
$$
\sum_{i\in\rV} {B^i_j}^*B^i_j=I
$$
for all $j\in\rV$. Hence, for all $j\in\rV$ there exists a unitary operator $U(j)$ on $\rH\otimes\rK$ whose first column (we choose $\vert 1\rangle$ to be the first vector) is given by
$$
U^i_1(j)=B^i_j\,.
$$
This unitary operator is a unitary operator that dilates the completely positive map  
$$
\rM_j(\r)=\sum_i B^i_j\,\r\, {B^i_j}^*\,.
$$
In other words, the completely positive map 
$\rM_j$
on $\rH$ is the partial trace of some unitary interaction between $\rH$ and some environment $\rE$. It is well-known that the dimension of the environment can be chosen to be the same as the number of Krauss operators appearing in the decomposition of $\rM_j$, that is, in our case they are indexed by $\rV$. Hence the environment can be chosen to be $\rE=\rK$. 

\smallskip
The state space on which one performs the realization procedure is $\rH\otimes\rK_1\otimes\rK_2$ where $\rK_1$ and $\rK_2$ are two copies of $\rK$. Let us present the main ingredients which shall appear in the realization procedure.

\smallskip
Each unitary operator $U(j)$ defined above acts on $\rH\otimes \rK_1$. We construct the unitary operator 
$$
U=\sum_j U(j)\otimes \vert j\rangle\langle j\vert
$$
which acts now on $\rH\otimes\rK_1\otimes\rK_2$\,. 

\smallskip
We shall also need the so-called \emph{swap operator} $S$ on  $\rK_1\otimes\rK_2$ defined by
$$
S (\vert j\rangle\otimes\vert k\rangle)=\vert k\rangle\otimes\vert j\rangle\,.
$$
It is a unitary operator on $\rK_1\otimes\rK_2$ which simply expresses the fact of exchanging the two systems $\rK_1$ and $\rK_2$.

\smallskip
We shall also use a \emph{decoherence procedure} on the space $\rK_1$, along the basis $(\vert i\rangle)$. By this we mean the following: if the system is in a superposition of pure states 
$$
\vert \varphi\rangle =\sum_i \l_i \, \vert i\rangle\,,
$$
then this system is coupled to an environment in such a way and in a sufficiently long time, for the state of $\rK_1$ to become
$$
\sum_i \ab{\l_i}^2\, \vert i\rangle\langle i\vert\,.
$$
This is to say that we have chosen a coupling of $\rK_1$ with some environment which makes the off-diagonal terms of the density matrix $\vert\varphi\rangle\langle\varphi\vert$ converge exponentially fast to 0. This kind of decoherence is now well-known in physics. It is rather easy to describe an environment and an explicit Hamiltonian which will produce such a result, we do not develop this point here, see \cite{A-P} for example.

\smallskip
Finally, we shall need a \emph{refreshing procedure}, that is, if $\rK_1$ is in any state $\r$ then we put it back to the state $\vert 1\rangle\langle 1\vert$. By this we mean either that the system $\rK_1$ is taken away  and re-prepared in the state $\vert 1\rangle\langle 1\vert$, or that a new copy of $\rK_1$ in the state $\vert 1\rangle\langle 1\vert$ is brought into the game in order to replace the old copy, which will play no role anymore (see \cite{A-P} for theoretical setup, or \cite{Har} for concrete experiments). 

\begin{proposition}\label{P:physical}
Consider the quantum system $\rH\otimes\rK_1\otimes\rK_2$, together with some initial state 
$$
\r^{(0)}=\sum_k \r_k\otimes \vert 1\rangle\langle 1\vert\otimes\vert k\rangle\langle k\vert\,.
$$
If we perform successively

\smallskip\noindent
1) an action of the unitary operator $U$,

\smallskip\noindent
2) a decoherence on the basis $(\vert i\rangle)$ of the system $\rK_1$,

\smallskip\noindent
3) an action of the swap operator $I\otimes S$

\smallskip\noindent
4) a refreshing of the system $\rK_1$ to the state $\vert 1\rangle\langle 1\vert$

\smallskip\noindent
then the state of the system becomes
$$
\sum_k \left(\sum_l B_l^k\r_l {B_l^k}^*\right)\otimes \vert 1\rangle\langle 1\vert\otimes\vert k\rangle\langle k\vert\,.
$$
That is, one reads the first step of the dissipative quantum random walk on $\rH\otimes\rK_2$. 

By iterating this whole procedure one produces the dissipative quantum random walk on $\rH\otimes\rK_2$. 
\end{proposition}
\begin{proof}
The unitary operator $U(k)$ admits a decomposition
$$
U(k)=\sum_{i,j} U^i_j(k)\otimes \vert j\rangle\langle i\vert\,.
$$
In particular we have
$$
\sum_j U^{i'}_j(k)^*\,U^i_j(k)=\delta_{i,i'} I\,.
$$
On the space $\rH\otimes\rK_1\otimes\rK_2$ the operator $U$ as defined above is then decomposed into
$$
U=\sum_{i,j,k} U^i_j(k)\otimes \vert j\rangle\langle i\vert\otimes \vert k\rangle\langle k\vert\,.
$$
Now starting in a pure state $\vert\phi\rangle\otimes\vert 1\rangle\otimes\vert k\rangle$, we get
$$
U\,(\vert\phi\rangle\otimes\vert 1\rangle\otimes\vert k\rangle)=\sum_j U^j_1(k)\vert\phi\rangle \otimes\vert j\rangle\otimes\vert k\rangle=\sum_j B^j_k\vert\phi\rangle \otimes\vert j\rangle\otimes\vert k\rangle\,.
$$
This is the first step of the procedure.

The second step consists in performing a decoherence on the first space $\rK$. The pure state 
$$
\sum_j B^j_k\vert\phi\rangle \otimes\vert j\rangle\otimes\vert k\rangle
$$
is then mapped to the density matrix
\begin{equation}\label{E:decoherent}
\sum_j B^j_k\vert\phi\rangle \langle \phi\vert{B^j_k}^*\otimes\vert j\rangle\langle j\vert\otimes\vert k\rangle\langle k\vert\,.
\end{equation}
 Applying $I\otimes S$ to the state (\ref{E:decoherent}) we get the state
\begin{equation}\label{E:final}
\sum_j B^j_k\vert\phi\rangle \langle \phi\vert{B^j_k}^*\otimes\vert k\rangle\langle k\vert\otimes\vert j\rangle\langle j\vert\,.
\end{equation}
On the space $\rH$ and the second space $\rK$ one can now read the first step of our quantum random walk.

Finally, \emph{refresh} the first space $\rK$ into the state $\vert 1\rangle$, we then end up with the state
$$
\sum_j B^j_k\vert\phi\rangle \langle \phi\vert{B^j_k}^*\otimes\vert 1\rangle\langle 1\vert\otimes\vert j\rangle\langle j\vert\,,
$$
on which one can apply our procedure again.

If the initial state is not a pure state but a density matrix, a mixture of pure states, it is not difficult to see that the procedure described above gives the right combination and the right final state.
\end{proof}

\smallskip
To summarize, the quantum random walk is obtained in the following way. Dilate each of the maps $\rL_k$ into a unitary operator $U(k)$ on $\rH\otimes \rK_1$, start in the desired initial state on $\rH$ and the second space $\rK_2$, with the first space $\rK_1$ being in the state $\vert 1\rangle$,   then iterate the following procedure on $\rH\otimes\rK_1\otimes\rK_2$:

\smallskip\noindent
-- apply the unitary operator $\sum_k U(k)\otimes \vert k\rangle\langle k\vert$,

\smallskip\noindent
-- perform a decoherence on  $\rK_1$

\smallskip\noindent
-- apply the unitary shift $I\otimes S$

\smallskip\noindent
-- refresh the first space $\rK_1$ into the state $\vert 1\rangle$. 

\smallskip\noindent
The dissipative quantum random walk now appears on $\rH\otimes\rK_2$. 

\section{Examples of Realization Procedure}

Let us illustrate this realization procedure with two of the physical examples developed in Sections \ref{S:exampleZ} and \ref{S:example_graph}. 

\smallskip
In the case of stationary walks on $\ZZ$ the procedure can be considerably simplified, as follows. The procedure we describe below is slightly different from the one presented in Proposition \ref{P:physical}, but it is actually the same one, presented in a different way, taking into account several simplifications offered by the model.

Consider an open quantum random walk on $\ZZ$ driven by two operators $B$ and $C$ on $\rH$. Consider a unitary operator $U$ on $\rH\otimes\CC^2$ of the form
$$
U=\left(
\begin{matrix} 
B&X\\
C&Y
\end{matrix}
\right)\,,
$$
that is, a dilation of the completely positive map driven by $B$ and $C$.
Let $\rK=\CC^\ZZ$ and consider the space $\rH\otimes\CC^2\otimes\rK$. 
On the space $\CC^2\otimes\rK$ we consider the \emph{shift} operator  given by
$$
S(\vert 0\rangle\langle 0\vert\otimes \vert k\rangle\langle k\vert)=\vert 0\rangle\langle 0\vert\otimes \vert k-1\rangle\langle k-1\vert
$$
and
$$
S(\vert 1\rangle\langle 1\vert\otimes \vert k\rangle\langle k\vert)=\vert 0\rangle\langle 0\vert\otimes \vert k+1\rangle\langle k+1\vert\,.
$$

Now, let us detail the procedure. Starting with a state $\vert\varphi\rangle\otimes \vert 0\rangle\otimes\vert k\rangle$ we apply the operator $U\otimes I$ and end up with the state
$$
B\vert\varphi\rangle\otimes \vert 0\rangle\otimes\vert k\rangle+C\vert\varphi\rangle\otimes \vert 1\rangle\otimes\vert k\rangle\,.
$$
Applying the decoherence on $\CC^2$ we get the state
$$
B\vert\varphi\rangle\langle \varphi\vert B^*\otimes \vert 0\rangle\langle 0\vert\otimes\vert k\rangle\langle k\vert+C\vert\varphi\rangle\langle \varphi\vert C^*\otimes \vert 1\rangle\langle 1\vert\otimes\vert k\rangle\langle k\vert\,.
$$
Applying the shift operator, the state becomes
$$
B\vert\varphi\rangle\langle \varphi\vert B^*\otimes \vert 0\rangle\langle 0\vert\otimes\vert k-1\rangle\langle k-1\vert+C\vert\varphi\rangle\langle \varphi\vert C^*\otimes \vert 1\rangle\langle 1\vert\otimes\vert k+1\rangle\langle k+1\vert\,.
$$
Refreshing the space $\CC^2$ we end up with
$$
B\vert\varphi\rangle\langle \varphi\vert B^*\otimes \vert 0\rangle\langle 0\vert\otimes\vert k-1\rangle\langle k-1\vert+C\vert\varphi\rangle\langle \varphi\vert C^*\otimes \vert 0\rangle\langle 0\vert\otimes\vert k+1\rangle\langle k+1\vert\,.
$$
One can read the first step of the dissipative quantum random walk on $\rH\otimes\rK$:
$$
B\vert\varphi\rangle\langle \varphi\vert B^*\otimes\vert k-1\rangle\langle k-1\vert+C\vert\varphi\rangle\langle \varphi\vert C^*\otimes\vert k+1\rangle\langle k+1\vert\,.
$$

Let us now detail the case of the first example of Section \ref{S:example_graph}, the open quantum random walk on the 2-vertices graph.

Consider a two-level quantum system $\rH$ coupled to another two-level quantum system $\rK_1$ via the  Hamiltonian
$$
H=i\gamma(\s^+\otimes \s^--\s^-\otimes \s^+)
$$
where 
$$
\s^+=\left(\begin{matrix} 0&1\\0&0\end{matrix}\right)\qquad\mbox{and}\qquad\s^-=\left(\begin{matrix} 0&0\\1&0\end{matrix}\right)\,.
$$
Then the unitary evolution associated to this Hamiltonian, for a time length $t=1$ (and $\hbar=1$) is given by
$$
U=e^{-iH}=\left(\begin{matrix}1&0&0&0\\ 0&\cos(\g)&-\sin(\g)&0\\0&\sin(\g)&\cos(\g)&0\\0&0&0&1\end{matrix}\right)\,.
$$
Hence, for a good choice of $\g$, that is, for $\sin(\g)=\sqrt p$ we have
$$
U=\left(\begin{matrix}1&0&0&0\\ 0&\sqrt{1-p}&-\sqrt p&0\\0&\sqrt p&\sqrt{1-p}&0\\0&0&0&1\end{matrix}\right)\,.
$$
In other words $U$ is of the form
$$
U=\left(\begin{matrix}C&X\\B&Y\end{matrix}\right)
$$
as a block matrix on $\rK_1$, where $B$ and $C$ are those matrices associated to our example. This is to say that we have given here an explicit dilation of the completely positive map associated to the matrices $B$ and $C$. 

\smallskip
If $D_1$ and $D_2$ are two diagonal matrices satisfying $D_1^*D_1+D_2^*D_2=I$ then assume, for simplicity only, that they have real entries
$$
D_1=\left(\begin{matrix}a&0\\0&\a\end{matrix}\right),\qq D_2=\left(\begin{matrix}b&0\\0&\b\end{matrix}\right)\,,
$$
with $a^2+b^2=\a^2+\b^2=1$. Then, one can write $a=\cos(\l)$ and $\a=\cos(\m)$. Considering the Hamiltonian
$$
K=\left(\begin{matrix} \l&0\\0&\m\end{matrix}\right)\otimes\left(\begin{matrix} 0&-i\\i&0\end{matrix}\right)\,,
$$
we get that $e^{-iK}$ is of the form
$$
V=\left(\begin{matrix} D_1&X'\\D_2&Y'\end{matrix}\right)\,.
$$
We have realized a concrete physical dilation of the completely positive map associated to $D_1$ and $D_2$. 

\smallskip
Following Proposition \ref{P:physical}, consider on $\rH\otimes\rK_1\otimes\rK_2=\CC^2\otimes\CC^2\otimes\CC^2$ the unitary evolution
$$
\left(\begin{matrix} V&0\\0&U\end{matrix}\right)\,,
$$
written as a block matrix on $\rK_2$. 
This is to say that $\rH$ is coupled to $\rK_1$ with the Hamiltonian $K$ when $\rK_2$ is in the state $\vert 1\rangle\langle 1\vert$ and $\rH$ is coupled to $\rK_1$ with the Hamiltonian $H$ when $\rK_2$ is in the state $\vert 2\rangle\langle 2\vert$. 

\smallskip
In this context, the swap operator $S$ takes the following simple form on $\rK_1\otimes\rK_2$
$$
S=\left(\begin{matrix}1&0&0&0\\0&0&1&0\\0&1&0&0\\0&0&0&1\end{matrix}\right)\,.
$$

\smallskip
Then, following the four steps of Proposition \ref{P:physical} gives a realization of the associated quantum random walk on $\rH\otimes\rK_2$. 

\section{Unitary Quantum Random Walks}\label{S:UQRW}

The Open Quantum Random Walks we have been describing up to now are actually very different from the well-known Unitary Quantum Random Walks, such as the Hadamard random walk (see Introduction for some references). This is to say that they produce probability distributions which are not of the same type as the ones usually observed with the Hadamard quantum random walks. 

It seems to us that there is no way to produce limit distributions such as the one observed in the Hadamard quantum random walk central limit theorem, with open quantum random walks. The limit behaviors of Open Quantum Random Walks seems to be all Gaussian or mixtures of Gaussians. The dissipative character of our quantum random walks makes them very different from the unitary evolution describing the usual type of quantum random walks. There is no inclusion, direct connection or simplification which establishes a direct link between OQRW and Unitary Quantum Random Walks.

However, there is quite a surprising and strong link between the two types of quantum random walks which appears via the physical realization procedure presented in Section 8. Under some conditions on the $B_i^j$'s, by modifying this procedure one can produce the usual unitary quantum walks. We insist again on the fact that this connection does not say that OQRW can give rise to Unitary Quantum Walks in some cases, but only that by modifying one step of the physical procedure, we get the unitary random walks. Let us develop all this here.

\smallskip
Let $\rV$ be a set of vertices, let $\rH$ be a Hilbert space representing the chirality. For each pair $(i,j)$ in $\rV^2$ we have a bounded operator $B^i_j$ on $\rH$. Instead of the usual condition
$$
\sum_i {B^i_j}^* B^i_j=I
$$ for all $j$, we shall ask here a much stronger condition, namely for all $j,j'\in\rV$
\begin{equation}\label{E:jjprime}
\sum_i {B^i_j}^*B^i_{j'}=\delta_{jj'} I\,.
\end{equation}
In other words, being given two starting points $j$ and $j'$, the sum of the ``contributions" which go to the same points $i\in\rV$ vanish, unless $j=j'$ in which case we recover the usual condition.

Note that there is no analogue of this condition for classical Markov matrices.

\smallskip
Let us illustrate this condition with an example. For a stationary quantum random walk on $\ZZ$ we are given two operators $B$ and $C$ on $\rH$ which represent the effect of making one step to the left or one step to the right. The usual condition, obtained by taking $j=j'$ gives
$$
B^*B+C^*C=I\,.
$$
Now, taking $j'=j+1$, we get a supplementary condition:
$$
C^*B=0\,.
$$
This is the only new condition added to the usual one in that case. 
Note that these two conditions together imply in particular that $B+C$ is unitary.

These two conditions are typically satisfied by the following family of examples. Let 
$$
U=\left(
\begin{matrix}
a&b\\
c&d
\end{matrix}
\right)
$$
be a unitary matrix on $\CC^2$. Put 
$$
B=\left(
\begin{matrix}
a&b\\
0&0
\end{matrix}
\right)\qquad\mbox{and}\qquad C=\left(
\begin{matrix}
0&0\\
c&d
\end{matrix}
\right)\,.
$$
Then, $B$ and $C$ satisfy 
$$
B^*B+C^*C=I\qquad\mbox{and}\qquad C^*B=0\,.
$$
This is typically the case with Hadamard random walk where
$$
U=\frac{1}{\sqrt 2}\left(
\begin{matrix}
1&1\\
1&-1
\end{matrix}
\right)\,.
$$

\smallskip
Let us see what happens, in the general context, with this additional condition. The point is the following, if we are given a pure state on $\rH\otimes\rK$ of the form
$$
\vert\psi\rangle=\sum_i \vert\varphi_i\rangle\otimes\vert i\rangle
$$
with the condition
$$
\normca{\psi}=\sum_i\normca{\varphi_i}=1
$$
then the state
$$
\vert\psi'\rangle=\sum_i\left(\sum_j B^i_j \vert \varphi_j\rangle\right)\otimes \vert i\rangle
$$
is of the same form and satisfies
\begin{align*}
\normca{\psi'}&=\sum_i \sum_{j,j'} \ps{\varphi_j}{{B^i_j}^*B^i_{j'}\, \varphi_{j'}}\\
&=\sum_{j,j'} \ps{\varphi_j}{\delta_{jj'}I\, \varphi_{j'}}\\
&=\sum_j \normca{\varphi_j}\\
&=1\,.
\end{align*}
Hence, at each step we get a state of the form 
$$
\vert\psi\rangle=\sum_i \vert\varphi_i\rangle\otimes\vert i\rangle
$$
with the condition
$$
\normca{\psi}=\sum_i\normca{\varphi_i}=1\,.
$$
In particular it determines, at each step, a probability distribution on $\rV$ by putting
$$
P(i)=\normca{\varphi_i}\,.
$$
This is exactly the picture for the Unitary Quantum Random Walks, such as the Hadamard quantum random walk. 

\smallskip
Now the interesting point is the way one can physically realize such Unitary Quantum Random Walks and the way this construction is similar to the one of Open Quantum Random Walks.

\begin{proposition}\label{P:physical2}
If the transition operators $B^i_j$ satisfy the more restrictive condition (\ref{E:jjprime}), then applying the same physical procedure as in Proposition \ref{P:physical} without the decoherence step (step 2) gives rise to a unitary quantum random walk.
\end{proposition}
\begin{proof}
Let us follow again the steps of the construction in Proposition \ref{P:physical}.  
Starting in a pure state $\vert\phi\rangle\otimes\vert 1\rangle\otimes\vert k\rangle$, we get
$$
U\,(\vert\phi\rangle\otimes\vert 1\rangle\otimes\vert k\rangle)=\sum_j U^j_1(k)\vert\phi\rangle \otimes\vert j\rangle\otimes\vert k\rangle=\sum_j B^j_k\vert\phi\rangle \otimes\vert j\rangle\otimes\vert k\rangle\,.
$$
This is the first step of the procedure.

We now skip the decoherence part and apply $I\otimes S$ to the state. We get the state
\begin{equation}\label{E:final}
\sum_j B^j_k\vert\phi\rangle \otimes\vert k\rangle\otimes\vert j\rangle\,.
\end{equation}
On the space $\rH$ and the second space $\rK$ one can now read the first step of the quantum random walk.

Finally, \emph{refresh} the first space $\rK$ into the state $\vert 1\rangle$, we then end up with the state
$$
\sum_j B^j_k\vert\phi\rangle \otimes\vert 1\rangle\otimes\vert j\rangle\,,
$$
on which one can apply our procedure again. We recognize the action of the type of quantum random walks we announced. 
\end{proof}

\end{document}